\documentclass{article}

\usepackage{setspace,url}
\usepackage[left=1in,top=1in,right=1in,bottom=1in,dvips,letterpaper]{geometry}
\usepackage[algo2e,linesnumbered,vlined,ruled]{algorithm2e}
\usepackage{graphics,graphicx,epsf,subfigure,epstopdf}
\usepackage{comment}
\usepackage{times}
\usepackage{amsmath,amsxtra,amsfonts,amscd,amssymb,bm}
\usepackage{amsthm}
\usepackage[usenames]{color}
\usepackage[normalem]{ulem}
\usepackage{exscale}

\newtheorem{lemma}{Lemma}[section]
\newtheorem{theorem}[lemma]{Theorem}

\newtheorem{condition}[lemma]{Condition}

\newcommand{\st}{\mbox{s. t.}}
\newcommand{\R}{\mathbb{R}}

\newcommand{\cL}{\mathcal{L}}

\newcommand{\tr}{\mbox{tr}}

\newcommand{\Rnn}{\R^{n \times n}}

\newcommand{\Rnk}{\R^{n \times k}}
\newcommand{\Rkk}{\R^{k \times k}}
\newcommand{\bX}{\bar{X}}
\newcommand{\bY}{\bar{Y}}

\newcommand{\diag}{{\rm diag}}

\newcommand{\Diag}{{\rm Diag}}
\newcommand{\dct}{\mathbf{d_{c2}}}
\newcommand{\dcts}{\mathbf{d^2_{c2}}}
\newcommand{\dctc}{\mathbf{d^3_{c2}}}
\newcommand{\dpt}{\mathbf{d_{p2}}}
\newcommand{\dpts}{\mathbf{d^2_{p2}}}

\newcommand{\zz}{^{\mathrm{T}}}

\newcommand{\inv}{^{-1}}
\newcommand{\ff}{_{\mathrm{F}}}

\newcommand{\mR}{\mathbb{R}}

\newcommand{\Lb}{\Lambda}

\newcommand{\half}{\frac{1}{2}}
\newcommand{\shalf}{\frac{\sqrt{2}}{2}}
\newcommand{\dX}{\Delta X}
\newcommand{\dY}{\Delta Y}

\newcommand{\dH}{\Delta H}
\newcommand{\tX}{\tilde X}

\newcommand{\be}{\begin{equation}}
\newcommand{\ee}{\end{equation}}
\newcommand{\bee}{\begin{equation*}}
\newcommand{\eee}{\end{equation*}}
\newcommand{\bea}{\begin{eqnarray}}
\newcommand{\eea}{\end{eqnarray}}
\newcommand{\beaa}{\begin{eqnarray*}}
\newcommand{\eeaa}{\end{eqnarray*}}

\begin{document}
\title{On the Convergence of the Self-Consistent Field Iteration in
Kohn-Sham Density Functional Theory}


\author{Xin Liu\thanks{State Key Laboratory of Scientific and Engineering Computing, Academy of Mathematics and Systems Science, Chinese Academy of
Sciences, CHINA (liuxin@lsec.cc.ac.cn). Research supported in part
by NSFC grants 11101409 and 11331012, and the National Center for
Mathematics and Interdisciplinary Sciences, CAS.}
\and Xiao Wang\thanks{School of Mathematical Sciences, University of Chinese
Academy of Sciences, CHINA (wangxiao@ucas.ac.cn). Research supported in part by Postdoc Grant 119103S175, UCAS president grant Y35101AY00, and NSFC grant 11301505.}
\and Zaiwen Wen\thanks{Beijing International Center for Mathematical
Research, Peking University, CHINA (wenzw@math.pku.edu.cn).
Research supported in part by NSFC grants 11101274, 11322109 and 91330202.}
\and  Yaxiang Yuan\thanks{State Key Laboratory of Scientific and Engineering Computing,
Academy of Mathematics and Systems Science, Chinese Academy of
Sciences, CHINA (yyx@lsec.cc.ac.cn). Research supported in part
by NSFC grant 11331012.}
}


\maketitle

{\footnotesize
\textbf{Abstract.} It is well known that the self-consistent field (SCF) iteration for
solving the  Kohn-Sham (KS) equation often fails to converge, yet there is no clear
explanation. In this paper, we investigate the SCF iteration from the perspective
of minimizing the corresponding KS total energy functional.  By analyzing  the
second-order Taylor expansion of the KS total energy functional and estimating the
relationship between the Hamiltonian and the part of the Hessian which is not
used in the SCF iteration, we are able to prove global
convergence from an arbitrary initial point and local linear convergence from an
initial point sufficiently close to the solution of the KS equation under assumptions that the gap between the occupied states and unoccupied
states is sufficiently large  and the second-order derivatives of the exchange
correlation functional are uniformly bounded from above. Although these
conditions are very stringent and are almost never satisfied in reality, our analysis is interesting in the sense that it provides a qualitative prediction of the behavior of the SCF iteration.  


}

\section{Introduction}

Consider the discretized Kohn-Sham (KS) equation
\be \label{eq:KS} \begin{aligned}  H(X) X &=  X \Lb, \\
  X\zz X &=  I,
  \end{aligned} \ee
where $X\in\Rnk$, the discretized Hamiltonian $H(X)\in\Rnn$ is a matrix function with respect to $X$ such that $H(X)X$ is equal to the
gradient of some discretized total energy functional $E(X)$ (to be defined in section \ref{sec:prob}), 
and $\Lb\in\Rkk$ is a diagonal matrix consisting of $k$ smallest eigenvalues of
$H(X)$. 
The discretized KS equation is a fundamental nonlinear eigenvalue problem arising from the density functional
theory (DFT) for electronic structure calculations
\cite{Martin2004,SzaboOstlund1996}, 
 in which the discretized charge density of electrons is defined as 
 \begin{eqnarray}\label{eq:rho}
\rho(X) := \diag(XX\zz),
\end{eqnarray}
where $\diag(A)$ denotes the vector containing the diagonal elements of the matrix $A$.
If no confusion can arise, we omit the word ``discretized" before ``KS''
and ``charge density'', etc.

The most
widely used approach for solving \eqref{eq:KS}
 is the
self-consistent field (SCF) iteration.
Starting from $X^0$ with $(X^0) \zz X^0=I$, the SCF iteration
computes the $(i+1)$-th iterate $X^{i+1}$ as the solution of the linear eigenvalue
problem:
\be \label{eq:SCF} \begin{aligned}  H(X^i) X^{i+1} &=  X^{i+1} \Lambda^{i+1}, \\
  (X^{i+1})\zz X^{i+1} &=  I.
  \end{aligned} \ee
When the difference between two consecutive Hamiltonians is negligible, the
system is said to be self-consistent and the SCF
procedure is terminated.  Heuristics have been proposed to accelerate and stabilize the SCF iteration. For
example, the charge mixing techniques \cite{Kerker1981,Kresse-1996} replace the Hamiltonian by a new matrix constructed
from a linear
combination of either the
potential or the charge densities computed in the previous SCF iterations and a
 new one obtained from certain schemes.

It is well known that the basic version of SCF iteration \eqref{eq:SCF} often converges slowly or fails to converge \cite{KoutechyBonacic1971} even with the help of various heuristics for decades, yet a clear explanation is not available.
 A convergence analysis  of the SCF
iteration for solving the Hartree-Fock equations according to the optimal
damping algorithm (ODA) is established in \cite{CancesLeBris2000}. The
interested reader is referred to 
 \cite{LeBris2005,Cances2000,Cances2001,Cancesetal2000,Cancesetal2003,Cancesetal2008,Kudinetal2002}
 on discussing ODA and its theoretical properties.
Recently, an analysis of gradient-based algorithms  for the Hartree-Fock equations
is proposed in \cite{Levitt2012} using Lojasiewiscz inequality. 
Some analysis on gradient-based algorithms can also be found in 
\cite{Schneideretall2009}.
In \cite{YangGaoMeza2009}, the authors prove that the sequence
generated by the SCF
iteration converges alternatively to two limit points
 which do not satisfy \eqref{eq:KS} on certain type of problems. A few  
 numerical explanations are provided in \cite{YangMezaWang2007}  
 by viewing the SCF iteration as an indirect
 procedure of minimizing a sequence of
quadratic surrogates.  A  condition is identified in \cite{YangGaoMeza2009}
to guarantee that
the SCF iteration becomes
a contractive fixed point iteration under a specific form of the Hamiltonian
without involving any exchange correlation term.
Basically, the condition characterizes the contribution of the nonlinear component
of the Hamiltonian.

In this paper, we establish some conditions on ensuring global and local
convergence of the SCF iteration for general Kohn-Sham DFT from an optimization point of view.
  Actually, the KS equation \eqref{eq:KS} is closely related to 
  the constrained minimization problem with orthogonality
constraints
\be\label{prob:minKS}  \begin{aligned} \min_{X \in \Rnk} & \quad   E(X)  \\
  \st \; \; &\quad X\zz X = I.
  \end{aligned} \ee
   The first-order optimality conditions for \eqref{prob:minKS} are  
   the same as \eqref{eq:KS} except that the diagonal  matrix $\Lb$ consists of
   any $k$ eigenvalues of $H(X)$ rather than the $k$ smallest ones.
 Assume that the second-order derivative of
  the exchange correlation energy functional is uniformly bounded from above, 
  which implies the Lipschitz
  continuity of the Jacobian of the functional.
Inspired by the expression of the
exact Hessian of $E(X)$ discovered in
 \cite{GaoYangMeza2009,Wen2013},  we observe that the  SCF
 iteration discards a ``complicate'' term in the Hessian of the total energy
 functional $E(X)$. Our analysis shows that this term plays an important role in
 the performance of the SCF scheme \eqref{eq:SCF}.  Briefly speaking, it
  converges if the gap between the $k$th and $(k+1)$st eigenvalues of the
 Hamiltonian $H(X)$ outweighs the norm of the complicate term in the
 Hessian up to some constant.  
 Although this condition is very stringent and is almost never satisfied in practice, which explains why the simplest SCF iteration often does not converge, our presented analysis is  
  interesting theoretically
 in the sense that it provides a qualitative prediction of the behavior of the SCF iteration with respect to the spectral gap of the nonlinear Hamiltonian relative to the Coulomb interaction.

The rest of this paper is organized as follows. In section \ref{sec:prob}, we describe the
total energy functional and its gradient and Hessian, as well as the distance
measurements between subspaces in detail.
The global and local convergence of the SCF iteration are presented in section
\ref{sec:global} and \ref{sec:local}, respectively.  Some relationship to the condition in
\cite{YangGaoMeza2009} is  clarified in section \ref{sec:comparison}.
Finally, we conclude our paper in the last section.

\section{Problem Statement} \label{sec:prob}

\subsection{The KSDFT Total Energy Functional}
Consider the discretized KS total energy
functional based on plane wave discretization  as
\be \label{eq:KS-energy} E(X):=\frac{1}{4}\tr(X\zz L X) + \half\tr(X\zz V_{ion}  X)
 + \half \sum_i \sum_l |x_i\zz w_l|^2 + \frac{1}{4} \rho^\top
L^\dagger \rho + \half e\zz \epsilon_{xc}(\rho), \ee
where $X=[x_1, \ldots, x_k] \in \Rnk$. 
The first term of \eqref{eq:KS-energy} is the so-called kinetic energy,
where  $L$ is a finite dimensional representation of the Laplacian operator. The
second term denotes local ionic potential energy, where the diagonal matrix
$V_{ion}$ is the ionic
pseudopotentials sampled on the suitably chosen Cartesian grid.  The third term  defines the nonlocal ionic potential
energy, where $w_l$ represents a discretized pseudopotential reference projection
function. The matrix $L^\dagger$
 corresponds to the pseudo-inverse of $L$ and the fourth term denotes the Hartree
potential energy, which is used to model the classical electrostatic average
interaction between electrons. The final term denotes the exchange correlation
energy, which is used to describe the nonclassical   interaction between
electrons. More detailed description of each terms of $E(X)$ can  be found in
\cite{YangMezaLeeWang2009,YangMezaWang2007}. Although the function  
\eqref{eq:KS-energy} is  can be different if other basis functions, such as Gaussian
atomic orbitals,   are used for the
discretization, our analysis still holds with some obvious modifications.

It can be verified that the gradient of $E(X)$ with respect to $X$ is
 $\nabla E(X) = H(X)X$,
where the Hamiltonian
 \be \label{eq:H} H(X):= \half L + V_{ion} + \sum_l w_l w_l\zz +
\Diag(  L^\dagger \rho) + \Diag( \mu_{xc}(\rho)\zz e),\ee
 and $ \mu_{xc}(\rho) :=\frac{ \partial \epsilon_{xc}}{\partial \rho} \in \R^{n \times
 n}$ and $\Diag(x)$ (with an uppercase letter $D$) denotes a diagonal matrix
with $x$ on its diagonal.
 Let $\cL(\Rnk,\Rnk)$ denote the space of linear operators which map $\Rnk$ to
$\Rnk$. The Fr\'echet derivative of $\nabla E(X)$ is defined as the (unique)
function $\nabla^2 E: \Rnk \to \cL(\Rnk, \Rnk)$ such that
\[
\lim_{\|S\|\ff \to 0 } \frac{\| \nabla E(X+S) -\nabla E(X) - \nabla^2 E(X)
(S)\|\ff }{\|S\|\ff } =0.
\]
The next lemma shows an explicit form of the Hessian operator
\cite{GaoYangMeza2009,Wen2013}.
 \begin{lemma}[Lemma 2.1 in \cite{Wen2013}] \label{lemma:Hessian}Suppose that $\epsilon_{xc}(\rho(X))$ is twice
   differentiable with respect to $\rho(X)$.
Given a direction $S \in \Rnk$, the Hessian-vector product of $E(X)$ is
\be \label{eq:secondorderdef} \nabla^2E(X)[S] =   H(X) S  +B(X)[S], \ee
where $ J :=  L^\dagger + \frac{\partial^2  \epsilon_{xc}}{ \partial \rho ^2 }
e$ and
\be\label{eq:secondpartdef}
B(X)[S] := 2 \Diag \left(J \diag(S  X \zz)  \right) X .
\ee
\end{lemma}

We make the following assumptions on the total energy function.
\begin{condition}\label{cond:exc}
The second-order derivatives of the exchange correlation functional
 $\epsilon_{xc}(\rho)$ is uniformly bounded from above, which implies the Lipschitz
   continuity of its Jacobian.
 Without loss of generality, 
 we assume that there exists a constant $\sigma$ such that 
 \[ \left\|\Diag(\mu_{xc}(\rho)\zz e) - \Diag(\mu_{xc}(\tilde \rho)\zz e)
 \right\|\ff \le \sigma \|\rho -
 \tilde \rho\|_2  \mbox{~ and ~}
 \left\| \frac{\partial^2  \epsilon_{xc}}{ \partial \rho ^2 } e \right\|_2 \le
 \sigma,\quad \mbox{ for all }\rho\in\mR^{n} . \]
\end{condition}

 We next consider the second part of the Hessian operator $B(X)[S]$ defined in \eqref{eq:secondpartdef}.

\begin{lemma}\label{lm:sndorder} Suppose that Condition \ref{cond:exc} holds.
  Let $X\in\mathcal{O}^{n\times k}$, $Z\in\mathcal{O}^{n\times (n-k)}$ and $S \in \Rnk$. Then 
\begin{eqnarray}\label{eq:secondpartdef1}
  \| B(X)[S]\|\ff &\leq& 2 \sqrt{n} (\|L^\dagger\|_2+ \sigma)\cdot \|S\|_2, \\
\label{eq:secondpartdef2}
\|Z\zz B(X)[ZZ\zz S]\|\ff &\leq& 2 \sqrt{n} (\|L^\dagger\|_2+ \sigma)\cdot \|Z\zz S\|_2.
\end{eqnarray}
\end{lemma}
\begin{proof}
We only prove the second inequality. Using $\|Z\zz\|_2 \le 1$ and $\|X\|_2=1$, we obtain
\beaa
\| Z\zz B(X)[ZZ\zz S]\|\ff &=& \|2 Z\zz \Diag(J \diag(ZZ\zz SX\zz))X\|\ff \nonumber\\
&\leq& 2\|Z\zz \|_2 \|\Diag(J \diag(ZZ\zz SX\zz))\|\ff\|X\|_2 \nonumber\\
&\leq& 2 \|\Diag(J \diag(ZZ\zz SX\zz))\|\ff = 2  \|J \diag(ZZ\zz SX\zz)\|_2 \nonumber\\
&\leq& 2 \|J\|_2\cdot \|\diag(ZZ\zz SX\zz)\|_2 \leq 2 \|J\|_2\cdot\sqrt{n} \|ZZ\zz SX\zz\|_\infty \nonumber\\
&\leq& 2 \sqrt{n} \|J\|_2\cdot \|ZZ\zz SX\zz\|_2 \leq 2 \sqrt{n} \|J\|_2\cdot
\|Z \zz S\|_2,
\nonumber
\eeaa
where the last inequality uses the fact that $\|Z M\|_2\leq \|M\|_2$ for any
matrix $M \in \R^{k \times k}$. This completes the proof.
\end{proof}

Our analysis also relies on the gap between the $k$th and $(k+1)$st eigenvalues of $H(X)$.

\begin{condition}
\label{cond:UWP}
Let  $\lambda_1 \le \ldots \le \lambda_k < \lambda_{k+1} \le \ldots \le
\lambda_n$ be the eigenvalues of a symmetric matrix $H \in \R^{n\times n}$.
There exists a gap between the $k$th and $(k+1)$st eigenvalues, that is, $\lambda_{k+1}-\lambda_k \ge \delta$ for some positive constant $\delta$.
\end{condition}
If Condition \ref{cond:UWP} holds for a sequence of matrices $\{H^i\}$
($i=1,2,...$) whose $\delta$ is uniformly bounded away from zero, $\{H^i\}$ is
said  to be uniformly well posed (UWP) in \cite{LeBris2005,YangGaoMeza2009}.

\subsection{Distance Measurements}
The SCF iteration maintains orthogonality in each iteration. The feasible set
\[ \mathcal{O}^{n\times k}:=\{X\mid X\in\Rnk, X\zz X=I\}\] is often referred to
as the Stiefel manifold.
The solutions of the KS equation \eqref{eq:KS}, the SCF iteration \eqref{eq:SCF}
and the minimization problem \eqref{prob:minKS} are invariant with respect
to orthogonal transformations. Namely, if $X$ is a solution,
 all points in the set $\{XU\mid U\in\mR^{k\times k},\,U\zz U=I_k\}$ are
 also solutions.
Hence, the Euclidean distance is not suitable to measure
the distance between a feasible point to a solution or a solution set of
\eqref{eq:KS}.  Inspired by the convergence analysis in \cite{YangGaoMeza2009}, we introduce two subspaces
distance measurements defined in section 4.3 of \cite{EdelmanAriasSmith1998} for
further analysis, i.e., for any $X_1,X_2\in\mathcal{O}^{n\times k}$,
\begin{enumerate}
\item {\bf Chordal 2-norm:\quad}  $\dct(X_1,X_2):=\min\limits_{Q_1,Q_2\in\mathcal{O}^{k\times k}} \|X_1Q_1-X_2Q_2\|_2$;
\item {\bf Projection 2-norm:\quad}  $\dpt(X_1,X_2):=\|X_1X_1\zz-X_2X_2\zz\|_2$.
\end{enumerate}
Let $U\Sigma V\zz$ be the singular value decomposition of $X_1\zz X_2$. It
 holds that
\begin{eqnarray}\label{eq:dcteq}
\dct(X_1,X_2)=\|X_1U-X_2V\|_2.
\end{eqnarray}

Since the equivalence
between $\dct$ and $\dpt$ is not discussed in
\cite{EdelmanAriasSmith1998}, we next include a proof for completeness.
\begin{lemma}\label{Rel:3}Given any $X_1,X_2\in\mathcal{O}^{n\times k}$, the Chordal 2-norm and Projection 2-norm satisfy
\begin{eqnarray}\label{eq:Rel}
\dct(X_1,X_2)  \geq \dpt(X_1,X_2) \geq \shalf\dct(X_1,X_2).
\end{eqnarray}
\end{lemma}
\begin{proof}
We first consider the first inequality in \eqref{eq:Rel}.
Let us denote $\bX_1 = X_1 U$ and $\bX_2 = X_2 V$, where $U$ and $V$ are defined in
\eqref{eq:dcteq}. Then, we observe
\begin{eqnarray}
0&\preceq& (I_k-\bX_1\zz \bX_2)(I_k-\bX_2\zz \bX_1)
= I - \bX_1\zz \bX_2 - \bX_2\zz \bX_1 + \bX_1\zz \bX_2\bX_2\zz \bX_1 \nonumber\\
& = &  (2I_k - \bX_1\zz \bX_2 - \bX_2\zz \bX_1) - (I_k -  \bX_1\zz \bX_2\bX_2\zz
\bX_1), \notag
\end{eqnarray}
which yields
\begin{eqnarray}\label{eq:2norm}
\sigma_{\max}(I_k -  \bX_1\zz \bX_2\bX_2\zz \bX_1) \leq
\sigma_{\max}(2I_k - \bX_1\zz \bX_2 - \bX_2\zz \bX_1).
\end{eqnarray}
Let $Z_2\in\mathcal{O}^{n\times (n-k)}$ be the orthogonal
complement to $X_2$. The left hand side of \eqref{eq:2norm} satisfies
\begin{eqnarray}\label{eq:lft}
\sigma_{\max}(I_k -  \bX_1\zz \bX_2\bX_2\zz \bX_1) &=&
\sigma_{\max}(\bX_1\zz(I_k-\bX_2\bX_2\zz)\bX_1)
= \sigma_{\max}(\bX_1\zz Z_2Z_2\zz \bX_1)\nonumber\\
& = & \|Z_2\zz \bX_1\|_2^2 = \dpts(\bX_1,\bX_2) = \dpts(X_1,X_2),
\end{eqnarray}
where the last equality holds due to Theorem 2.6.1 of \cite{Golub1996}.
 It follows from \eqref{eq:dcteq} that the right hand side of \eqref{eq:2norm} satisfies
\begin{eqnarray}\label{eq:rgt}
\sigma_{\max}(2I_k - \bX_1\zz \bX_2 - \bX_2\zz \bX_1) = \|\bX_1 - \bX_2\|_2^2
= \dcts(X_1,X_2),
\end{eqnarray}
which together with \eqref{eq:lft} proves the first part of \eqref{eq:Rel}.

We now prove the second inequality of \eqref{eq:Rel}.
According to \eqref{eq:lft} and the definitions of $U$ and $V$, we obtain
\begin{eqnarray}\label{eq:lft2}
\dpts(X_1,X_2) = \sigma_{\max}(I_k -  \bX_1\zz \bX_2\bX_2\zz \bX_1)
= \sigma_{\max} (I_k - \Sigma^2).
\end{eqnarray}
It follows from \eqref{eq:rgt} that
\begin{eqnarray}\label{eq:rgt2}
\dcts(X_1,X_2) = \sigma_{\max}(2I_k - \bX_1\zz \bX_2 - \bX_2\zz \bX_1)
= \sigma_{\max} (2I_k - 2\Sigma).
\end{eqnarray}
Since $X_1$ and $X_2$ are orthogonal matrices, each diagonal entry of the
diagonal matrix $\Sigma$ is in $[0,1]$. The proof is completed by combining \eqref{eq:lft2}
and \eqref{eq:rgt2} together.
\end{proof}

Theorem 4.11 in \cite{Stewart1973} and Corollary 7.2.5 in \cite{Golub1996}
 are sufficient to guarantee the convergence of the invariant subspaces corresponding
to the $k$-smallest eigenvalues.
\begin{lemma}\label{lm:sensitivity}
Suppose that the symmetric matrix $H\in\mR^{n\times n}$ satisfies Condition
\ref{cond:UWP}. Let 
$\dH\in \mR^{n\times n}$ be a symmetric perturbation to $H$ and  $X,\tX\in\mR^{n\times
k}$ be the invariant subspaces
 associated with the $k$ smallest eigenvalues of $H$ and $H+\dH$, respectively. If $||\dH||_2$ is
 sufficiently small, it holds that
 \begin{eqnarray}
 \dpt(X,\tX) \leq C\cdot ||\dH||_2,
 \end{eqnarray}
 where $C$ is a parameter only related to $\delta$ in Condition
\ref{cond:UWP}.
\end{lemma}


\section{Global Convergence of the SCF Iteration}
\label{sec:global}
In this section, we prove global convergence of the SCF iteration based on the
reduction of the total energy functional between two consecutive iterates.
 Suppose that $X\in\mathcal{O}^{n\times k}$
 is an arbitrary feasible point of \eqref{prob:minKS}, and  $Y$ is
  obtained from
 running one SCF iteration with $X$ as the starting point. Namely, the columns of $Y$ are the eigenvectors associated with the $k$
smallest eigenvalues of $H(X)$.
Such a $Y$ is not unique because the linear eigenvalue problem is invariant with respect to
the orthogonal transformation.
Let $U\Sigma V\zz$ be the singular value
decomposition of $X\zz Y$, where $U,V\in\mathcal{O}^{k\times k}$. Then it follows from \eqref{eq:dcteq} that
$\bY:= YVU\zz$ satisfies
\begin{eqnarray}\label{eq:YChordaleq}
\|X-\bY\|_2 = \dct(X,Y).
\end{eqnarray}
Due to the invariance, 
$\bY$ is also a solution to the linear eigenvalue problem in the SCF iteration starting from $X$ and  $E(Y)=E(\bY)$.
For simplicity of notation, we call $\bY$ as the closest SCF iterate obtained
from $X$ under the Chordal 2-norm.

The second-order Taylor expansion of $E(Y)$ at $X$ gives
\begin{eqnarray}\label{eq:Y-X1}
E(Y) & = & E(X) + \langle \nabla E(X), Y-X \rangle + \frac{1}{2}\langle \nabla^2 E(D_t)[Y-X], Y-X \rangle \notag ,
\end{eqnarray}
where $D_t = X + t(Y-X)$ for some $t\in(0,1)$, and 
 the Euclidean inner product $\langle A_1,A_2\rangle$ between any real matrices
 $A_1,A_2\in\mR^{n\times k}$ is defined as $\tr(A_1^TA_2)$. Using the formulations of the
gradient $\nabla E(X)=H(X)X$ and the Hessian-vector product \eqref{eq:secondorderdef}, we obtain
\bea
 E(X) - E(Y)
&=&   -\langle \nabla E(X), Y-X \rangle - \frac{1}{2}\langle \nabla^2
E(X)[Y-X], Y-X \rangle \nonumber \\
 & & -\frac{1}{2}\langle \nabla^2 E(D_t)[Y-X], Y-X \rangle + \frac{1}{2}\langle
 \nabla^2 E(X)[Y-X], Y-X \rangle \nonumber \\
\label{eq:Y-X2}
 &=&   \frac{1}{2}(\langle H(X)X, X \rangle - \langle H(X)Y,Y \rangle )
 -R^{(1)}_X(Y,D_t) - R^{(2)}_X(Y,D_t),
 \eea
where
 \bea\label{eq:Y-X3}
R^{(1)}_X(Y,D_t) &:=&\frac{1}{2}\langle (H(D_t)-H(X))(Y-X), Y-X \rangle,\\
\label{eq:Y-X4}
R^{(2)}_X(Y,D_t) &:=& \frac{1}{2}\langle B(D_t)[Y-X], Y-X \rangle.
\eea
The first term of the right hand side in \eqref{eq:Y-X2} corresponds to a
reduction of a quadratic form of the linear eigenvalue problem in the SCF iteration.  Lemma 1 in \cite{YangGaoMeza2009} ensures
the following reduction.
\begin{lemma}\label{lm:citation}
  Suppose that Condition \ref{cond:UWP} holds at $H(X)$, and $Y$ is a solution
  obtained from running one SCF iteration with $X$ as the starting point. Then we have
\begin{eqnarray}\label{eq:YangIneq}
\langle H(X)X, X \rangle - \langle H(X)Y,Y \rangle   \ge \delta \cdot \dpts(X,Y).
\end{eqnarray}
\end{lemma}

We next estimate $R^{(1)}_X(Y,D_t)$ and
$R^{(2)}_X(Y,D_t)$ for the reduction of $E(X) - E(Y)$.

\begin{lemma}\label{lm:Red}
Suppose that Condition \ref{cond:exc} holds. 
Let $X$ be an orthogonal matrix with $H(X)$ satisfying Condition \ref{cond:UWP},
and $Y$ be a solution obtained from running one SCF iteration with $X$ as the starting point. Then
\begin{eqnarray}\label{bound_EX_EY}
E(X) - E(Y) &\geq& \frac{1}{2}\delta\cdot \dpts(X,Y)
 -  k\sqrt{n}( \|L^\dagger\|_2 + \sigma)\cdot (\dcts(X,Y) + \dctc(X,Y)).
\end{eqnarray}
\end{lemma}
\begin{proof}
 Let $\bY$ be the closest SCF iterate obtained from $X$ under the Chordal 2-norm.
Using the facts that the second term of the left hand side in
\eqref{eq:YangIneq} is invariant with respect to orthogonal transformation on $Y$ and $\dpt(X,Y)=\dpt(X,\bY)$, we obtain
\begin{eqnarray}\label{eq:Q}
\langle H(X)X,X \rangle - \langle H(X)\bY,\bY \rangle \;\ge \; \delta \cdot\dpts(X,\bY).
\end{eqnarray}

 Simple calculations show that
\begin{eqnarray}\label{eq:tobeimproved}
\|XX\zz-D_tD_t\zz\|_2 
                    &  \le &  2\|X - D_t\|_2 \;\le \; 2\|\bY-X\|_2.
\end{eqnarray}
The definition of $H(X)$, Condition \ref{cond:exc} and the inequality \eqref{eq:tobeimproved}
give
\bea
 && \|H(D_t) - H(X)\|_F \notag \\
& = & \|\mbox{Diag}(L^\dagger(\rho(X) - \rho(D_t)))\|_F +
\|\Diag( \mu_{xc}(\rho(X))\zz e) -\Diag( \mu_{xc}(\rho(D_t))\zz e) \|_F \notag\\
                    & \le & ( \|L^\dagger\|_2 + \sigma) \|\rho(X) - \rho(D_t)\|_2 \notag \\
                    & \le & \sqrt{n} ( \|L^\dagger\|_2 + \sigma)  \|\mbox{diag}(XX\zz) - \mbox{diag}(D_tD_t\zz)\|_{\infty} \notag \\
                    & \le & \sqrt{n} ( \|L^\dagger\|_2 + \sigma)  \|XX\zz - D_tD_t\zz\|_2 \notag \\
                    & \le & 2\sqrt{n} ( \|L^\dagger\|_2 + \sigma)  \|\bY-X\|_2,
                    \notag \label{H_xi}
\eea
which further yields
\begin{eqnarray}\label{eq:R1}
R^{(1)}_X(\bY, D_t) &\leq& \left|\frac{1}{2}\langle (H(D_t)-H(X))(\bY-X), \bY-X\rangle
\right| \nonumber\\
&\leq& \frac{1}{2} \|H(D_t)-H(X)\|\ff \|\bY-X\|_2 \|\bY-X\|\ff \nonumber\\
&\leq& k\sqrt{n}( \|L^\dagger\|_2 + \sigma) \|\bY-X\|_2^3.
\end{eqnarray}

It follows from \eqref{eq:secondpartdef1} in Lemma \ref{lm:sndorder} that
\beaa
\langle B(D_t)[\bY-X], \bY-X\rangle
& \le & \|B(D_t)[\bY-X]\|\ff\|\bY-X\|\ff \nonumber \\
& \le & 2 \sqrt{n}\|J\|_2\|D_t(\bY-X)\zz\|_2\cdot k\cdot\|\bY-X\|_2 \nonumber \\
& \le & 2 k \sqrt{n}( \|L^\dagger\|_2 + \sigma) \|\bY-X\|_2^2,
\eeaa
where the last inequality is implied by $\|D_t\|_2 = \|X + t(\bY-X)\|_2\le 1$.
Consequently, we have
\begin{eqnarray}\label{eq:R2}
R^{(2)}_X(\bY, D_t)\leq \left|\frac{1}{2}\langle B(D_t)[\bY-X], \bY-X \rangle
\right|\;\le\;  k \sqrt{n} ( \|L^\dagger\|_2 + \sigma)  \|\bY-X\|_2^2.
\end{eqnarray}

Substituting \eqref{eq:Q}, \eqref{eq:R1} and \eqref{eq:R2} into \eqref{eq:Y-X2}, we obtain
\begin{eqnarray}
E(X) - E(\bY) &\geq& \frac{1}{2}\delta\cdot \dpts(X,\bY)
 -  k\sqrt{n} ( \|L^\dagger\|_2 + \sigma)
( \|X-\bY\|_2^2 + \|X-\bY\|_2^3).
\end{eqnarray}
Finally, the inequality \eqref{bound_EX_EY} is proved by using \eqref{eq:YChordaleq}, $\dpt(X,Y)=\dpt(X,\bY)$ and $E(Y)=E(\bY)$.
\end{proof}


We now present our global convergence results based on the reduction of the total energy functioanl in Lemma \ref{lm:Red} and the relationship between
the distance measurements in Lemma \ref{Rel:3}.

\begin{theorem}\label{thm:global}
Suppose that Condition \ref{cond:exc} holds. 
Let $\{X^i\}$ be a sequence generated by the SCF iteration such that
$\{H(X^i)\}$ is uniformly well posed with a constant $\delta$. 
Then $\{X^i\}$ converges to a solution to the KS equation  \eqref{eq:KS}, if
\begin{equation}\label{bound_alpha}\delta > 12 k\sqrt{n}( \|L^\dagger\|_2 +
\sigma).
\end{equation}
\end{theorem}
\begin{proof}\quad It follows from Lemma \ref{Rel:3} and Lemma \ref{lm:Red}
  that, for any $i=1,2,...$,
\begin{eqnarray}\label{eq:gthm1}
E(X^{i}) - E(X^{i+1}) & \ge & \left(\frac{1}{4}\delta - k\sqrt{n}( \|L^\dagger\|_2 + \sigma)
\right)\dcts(X^{i},X^{i+1}) \notag \\
 & & -  k \sqrt{n}( \|L^\dagger\|_2 + \sigma)
 \dctc(X^i,X^{i+1}).
\end{eqnarray}
Since $X^{i}$ and $X^{i+1}$ are both orthogonal matrices, we have
\begin{eqnarray}\label{eq:gthm2}
\dct(X^i,X^{i+1}) \le \|X^i\|_2 + \|X^{i+1}\|_2 =2.
\end{eqnarray}
Substituting \eqref{eq:gthm2} into \eqref{eq:gthm1}, we obtain
\begin{equation}\label{red}
E(X^i) - E(X^{i+1})\; \ge \; (\frac{1}{4}\delta - 3k\sqrt{n}( \|L^\dagger\|_2 +
\sigma))\dcts(X^i,X^{i+1}).
\end{equation}
By summing  \eqref{red} over all indices from $0$ to $i$, we obtain
\be \label{eq:sum-E} E(X^{i+1}) \le E(X^0) - (\frac{1}{4}\delta - 3k\sqrt{n}( \|L^\dagger\|_2 +
\sigma)) \sum_{j=0}^{i} \dcts(X^i,X^{i+1}). \ee
Since
$E(X^i)$ is bounded  below, we have that $E(X^0)-E(X^{i+1}) $ is less than some
 positive constant for all $i$. Hence, by taking limits in \eqref{eq:sum-E}, we
 have 
\begin{eqnarray}\label{eq:gthmfinal}
 \lim\limits_{i\rightarrow\infty} \dct(X^i,X^{i+1})  = 0.
\end{eqnarray}
Namely, $\{X^i\}$ converges.\; Let
\be
X^*:=\lim\limits_{i\rightarrow\infty} X^i,
\ee
and $\tX$ be consisted of the eigenvectors associated with the $k$
smallest eigenvalues of $H(X^*)$.
It follows from Lemma \ref{lm:sensitivity} that
\begin{eqnarray}
\dpt(X^{i+1},\tX) \leq C\cdot ||H(X^i)-H(X^*)||_2.
\end{eqnarray}
Taking limit on both sides and using the continuity of $H(X)$,
we obtain
\begin{eqnarray}
0\leq \dpt(X^*,\tX) = \lim\limits_{i\rightarrow\infty}\dpt(X^{i+1},\tX)
\leq \lim\limits_{i\rightarrow\infty} C\cdot ||H(X^i)-H(X^*)||_2 = 0.
\end{eqnarray}
Namely, $X^* = \tX$, which completes the proof.
\end{proof}

Theorem \ref{thm:global} guarantees the convergence of the SCF iteration to
a solution of the KS equation, which is more than the first-order optimality conditions for \eqref{prob:minKS}. In fact, when the inequality \eqref{bound_alpha} holds, 
the reduction of the total energy \eqref{red} implies that any global minimizer of \eqref{prob:minKS} is a solution
of the KS equation.

\section{Local Convergence of the SCF Iteration} \label{sec:local}

In this section, we establish local convergence of the SCF iteration by
exposing the relationship between two consecutive iterates in terms of
their distances to a particular solution of \eqref{eq:KS}. The
results are called local analysis since it
relies on the Taylor expansion in a small neighborhood of
that optimal solution.

\begin{lemma}\label{order}
Suppose that Conditions \ref{cond:exc} holds.
Let $X^*$ be a solution to the KS equation \eqref{eq:KS} whose $H(X^*)$
satisfies Condition \ref{cond:UWP},
$X\in\mathcal{O}^{n\times k}$ be in a sufficiently small neighborhood of $X^*$, and $Y$
be a solution obtained from running one SCF iteration with $X$ as the starting point. 
Then $\dpt(X^*,Y)$ is of the same order of $\dpt(X^*,X)$, namely
\begin{eqnarray} \label{eq:equiv-d}
\dpt(X^*,Y)=O(\dpt(X^*,X)).
\end{eqnarray}
\end{lemma}
\begin{proof}
Using the continuity of $H(X)$, 
 the fact that $X$ is in a sufficiently small neighborhood of $X^*$ and Lemma
 \ref{lm:sensitivity}, we obtain 
\begin{eqnarray}\label{eq:InvErr}
\dpt(X^*,Y)\leq C\cdot ||H(X)-H(X^*)||_2 = O(||X-X^*||_2),
\end{eqnarray}
which proves \eqref{eq:equiv-d}.
\end{proof}

\begin{theorem}\label{thm:local}
Suppose that Conditions \ref{cond:exc} holds.
Let $X^*$ be a solution to the KS equation \eqref{eq:KS} whose $H(X^*)$
satisfies Condition \ref{cond:UWP},
$X$ be in a sufficient small neighborhood of $X^*$, 
and $Y$ be a solution obtained from running one SCF iteration with $X$ as the starting point.
Then
\begin{eqnarray}
\dpt(X^*,Y) \leq \frac{2 \sqrt{n} (\|L^\dagger\|_2+ \sigma)}{\delta}\cdot \dpt(X^*,X) + O(\dpts(X^*,X)).
\end{eqnarray}
\end{theorem}
\begin{proof}
For convenience of exposition, we introduce $\dX := X^*-X$ and $\dY := X^* - Y$.
Recalling the fact that $\nabla E(X) = H(X)X$, we obtain the first-order Taylor expansion of
$\nabla E(X^*)$ at $X$ as follows,
\bea\label{eq:dXdY1}
H(X^*)X^* = \nabla E(X^*) &=&  \nabla E(X) + \nabla^2 E(X)[\dX] + O(\|\dX\|_2^2)\nonumber\\
&=& H(X)X + H(X)\dX + B(X)[\dX] + O(\|\dX\|_2^2)\nonumber\\
&=& H(X)Y + H(X)\dY + B(X)[\dX] + O(\|\dX\|_2^2).
\eea
 Using Lemma \ref{order} and substituting $X^*$ by $Y+\dY$, we have
\bea\label{eq:dXdY2}
X^*(X^*)\zz H(X^*)X^*
&=& (Y+\dY)(Y+\dY)\zz(H(X)Y + H(X)\dY + B(X)[\dX] + O(\|\dX\|_2^2)) \nonumber\\
&=& YY\zz H(X)Y + Y\dY\zz H(X)Y + \dY Y\zz H(X)Y  \nonumber\\
&& + YY\zz H(X)\dY + YY\zz B(X)[\dX] + O(\|\dX\|_2^2).
\eea
By using the fact that $X^*$ is a global solution of \eqref{eq:KS} and $Y$ is
an SCF iterate obtained from $X$, we have
\begin{eqnarray}\label{eq:First1}
H(X^*)X^*&=& X^*(X^*)\zz H(X^*)X^*,\\
\label{eq:First2}
H(X)Y&=& YY\zz H(X)Y.
\end{eqnarray}
It follows from the relations \eqref{eq:dXdY1}-\eqref{eq:First2} that
\begin{eqnarray}\label{eq:final-rel1-1}
&& H(X)\dY -  ( Y\dY\zz H(X)Y + \dY Y\zz H(X)Y +
YY\zz H(X)\dY ) \nonumber\\
&=& - (I-YY\zz )B(X)[\dX] + O(\|\dX\|_2^2).
\end{eqnarray}
Consequently, the above relation and Lemma \ref{order} imply that 
\begin{eqnarray}\label{eq:final-rel1-2}
&& H(X^*)\dY -  ( X^*\dY\zz H(X)Y + \dY (X^*)\zz H(X^*)X^* +
X^*Y\zz H(X)\dY ) \nonumber\\
&=& - (I-X^*(X^*)\zz )B(X)[\dX] + O(\|\dX\|_2^2).
\end{eqnarray}
Let $Z^*$ be the orthogonal complement to $X^*$.
Multiplying both sides of \eqref{eq:final-rel1-2} with $(Z^*)^\top$ yields:
\begin{eqnarray}\label{eq:final-rel1}
&& (Z^*)\zz  H(X^*)\dY -  (Z^*)\zz( X^*\dY\zz H(X)Y + \dY (X^*)\zz H(X^*)X^* +
X^*Y\zz H(X)\dY ) \nonumber\\
&=& - (Z^*)\zz B(X)[\dX] + (Z^*)\zz X^*(X^*)\zz B(X)[\dX] + O(\|\dX\|_2^2),
\end{eqnarray}
which can be rewritten as
\begin{eqnarray}\label{eq:final-rel2}
(Z^*)\zz  H(X^*)\dY -  (Z^*)\zz \dY  (X^*)\zz H(X^*)X^* = - (Z^*)\zz B(X)[\dX] + O(\|\dX\|_2^2).
\end{eqnarray}

Let $\Lb_k$ and $\Lb_{n-k}$ be the diagonal matrices consisting of the $k$ smallest and
$n-k$ largest eigenvalues of $H(X^*)$, respectively.
It follows from \eqref{eq:First1} and the definition
of $Z^*$ that
\begin{eqnarray}\label{eq:final-rel3}
\Lb_{n-k}(Z^*)\zz\dY -  (Z^*)\zz \dY \Lb_k =
- (Z^*)\zz B(X)[(Z^*(Z^*)\zz+X^*(X^*)\zz)\dX] + O(\|\dX\|_2^2).
\end{eqnarray}
 By using the orthogonality of $X$, we have $(X^*-\dX)\zz (X^*-\dX) =X\zz X = I$,
which further gives,
\begin{eqnarray}\label{eq:dxxorder}
(X^*)\zz\dX =  O(\|\dX\|^2).
\end{eqnarray}
It follows from \eqref{eq:dxxorder} that
\begin{eqnarray}\label{eq:final-rel3-1}
\Lb_{n-k}(Z^*)\zz\dY -  (Z^*)\zz \dY \Lb_k =
- (Z^*)\zz B(X)[Z^*(Z^*)\zz\dX] + O(\|\dX\|_2^2).
\end{eqnarray}
Taking Frobenius-norm on both sides of \eqref{eq:final-rel3-1},
we have
\begin{eqnarray}\label{eq:final-rel4}
\|\Lb_{n-k} (Z^*)\zz\dY\|\ff - \| (Z^*)\zz \dY \Lb_k\|\ff \leq
 \| (Z^*)\zz B(X)[Z^*(Z^*)\zz\dX]\|\ff  + O(\|\dX\|_2^2).
\end{eqnarray}
Condition \ref{cond:UWP} implies
\begin{eqnarray}\label{eq:final-rel4-1}
\|\Lb_{n-k} (Z^*)\zz\dY\|\ff - \| (Z^*)\zz \dY \Lb_k\|\ff \geq \delta \|(Z^*)\zz\dY\|\ff.
\end{eqnarray}
By using Lemma \ref{lm:sndorder} and substituting \eqref{eq:final-rel4-1} into \eqref{eq:final-rel4},
we obtain
\begin{eqnarray}\label{eq:final-rel5}
\delta \|(Z^*)\zz\dY\|\ff  \leq
2 \sqrt{n} \|J\|_2\cdot \|(Z^*)\zz\dX \|_2 + O(\|\dX\|_2^2).
\end{eqnarray}
It is clear that $\dpt(X^*,Y) =\|(Z^*)\zz\dY\|_2\leq \|(Z^*)\zz\dY\|\ff$ and $\dpt(X^*,X) =\|(Z^*)\zz\dX\|_2$.
Recalling \eqref{eq:dxxorder} and the definition of $Z^*$,
we obtain
\begin{eqnarray}
||\dX||_2\geq \|(Z^*)\zz\dX\|_2 \geq ||\dX||_2 - ||(X^*)\zz\dX||_2 = ||\dX||_2 - O(||\dX||_2^2).
\end{eqnarray}
Namely,
 $O(\|\dX\|_2) = O(\dpt(X^*,X))$ holds, which completes the proof.
\end{proof}

Hence, when $2 \sqrt{n} (\|L^\dagger\|_2+ \sigma)<\delta$ holds, Theorem
\ref{thm:local} implies that the SCF iteration converges linearly to the
solution $X^*$ of the KS equation once the sequence locates in a sufficiently small
neighborhood of $X^*$.

\section{Comparison with the Results of Yang et al. in \cite{YangGaoMeza2009}}
\label{sec:comparison}
In this section, we explain the difference between our convergence results and these of
Yang et al. \cite{YangGaoMeza2009} on a special form of the total energy
functional as
\[ E(X) := \frac{1}{2} \tr(X\zz L X) +\frac{\alpha}{4}\rho(X)\zz L\inv \rho(X),  \]
whose Hamiltonian is
\[ H(X) := L +\alpha \Diag(L\inv \rho(X)). \]
Since there is no exchange correlation energy functional in this case, the
constant $\sigma=0$ in Condition \ref{cond:exc}.

Theorem \ref{thm:global} provides global convergence from any initial point if
\begin{equation}\label{bound_alpha-global}
\alpha \; < \; \alpha_G \; := \; \frac{\delta}{12k\sqrt{n}\|L^{-1}\|_2}.
\end{equation}
According to Theorem \ref{thm:local}, the SCF iteration converges linearly to
the optimal solution from an initial
 point located in a neighborhood of that solution,
 if $\alpha$ satisfies
\begin{equation}\label{bound_alpha-local}
\alpha \; < \; \alpha_L \; := \; \frac{\delta}{2\sqrt{n}\|L^{-1}\|_2}.
\end{equation}

On the other hand, Yang et al.  \cite{YangGaoMeza2009} proves convergence of a
variant of the SCF iteration
 whose the density function is computed by
\[  \rho= \diag(f_\mu(H)). \]
Here $f_{\mu}(t) := \frac{1}{1+e^{\beta(t-\mu)}}$  and $f_{\mu}(H) := V
\Diag(f_\mu(\lambda_1), \ldots, f_\mu(\lambda_n))V\zz$, where
$H=V\Diag(\lambda_1,\ldots,\lambda_n)V\zz$ is the eigenvalue decomposition
of $H$.
They provide global linear
convergence if
\begin{equation}\label{bound_alphaF}
\alpha\;< \; \alpha_F \; := \; \frac{2}{n^4\beta\|L^{-1}\|_1},
\end{equation}
where $\beta$ and $\mu$ satisfy
 \[ \mathrm{trace}(f_\mu(H)) = k. \]
For a given constant $\gamma \ll 1$, the smoothing can be achieved by requiring
\[\label{bond}
\begin{cases}
\frac{1}{1+e^{\beta(\lambda_k-\mu)}}\geq 1-\gamma,\\
\frac{1}{1+e^{\beta(\lambda_{k+1}-\mu)}}\leq \gamma,
\end{cases}
\]
which is equivalent to
\[\label{beta1}
\beta\;\ge \; \max\left\{\frac{\ln \frac{1-\gamma}{\gamma}}{\mu-\lambda_k}, \frac{\ln \frac{1-\gamma}{\gamma}}{\lambda_{k+1}-\mu}\right\}.
\]
Notice that
\[
\min\limits_{\mu}\max\left\{\frac{\ln \frac{1-\gamma}{\gamma}}{\mu-\lambda_k}, \frac{\ln \frac{1-\gamma}{\gamma}}{\lambda_{k+1}-\mu}\right\} = \frac{2}{\delta}\cdot \ln \frac{1-\gamma}{\gamma},
\]
whose minimum is achieved at $\mu=\frac{\lambda_k + \lambda_{k+1}}{2}$.
Therefore, we obtain
$\beta\geq \frac{2}{\delta}\cdot\ln\frac{1-\gamma}{\gamma}$.  Namely,
\begin{eqnarray}\label{bound_alpha3}
\alpha_F < \frac{\delta}{\ln\frac{1-\gamma}{\gamma}\cdot n^4\|L^{-1}\|_1}.
\end{eqnarray}

We notice that
$k\sqrt{n} < n^{1.5} < n^4$ and $k\sqrt{n} \ll n^4$ when $n$ is sufficiently large.
Moreover, $\ln\frac{1-\gamma}{\gamma} > 12$  if $\gamma < 6.1442
\times 10^{-6}$, whereas   $\ln\frac{1-\gamma}{\gamma}\cdot n^4 > 12 k\sqrt{n}$, when $\gamma < 0.1070$ and $n\geq 2$.
 By comparing \eqref{bound_alpha3} to \eqref{bound_alpha-global}, we can obtain that
$\alpha_F<\alpha_G$ under a reasonable  value of $\gamma$.
Furthermore, $\alpha_F\ll\alpha_G$ holds when $n$ is sufficiently large. Hence, we can conclude that
our condition  is no more restricted than the one in \cite{YangGaoMeza2009}.

\section{Conclusion}
We study the convergence issues of the well-known self-consistent
field (SCF) iteration for
solving the Kohn-Sham equation in density functional theory.  Our analysis is based on the
second-order Taylor expansion of the total energy functional. We show that a
``complicate'' part of the Hessian plays an important role in ensuring the convergence of
the SCF iteration.
 Both global and local convergence can be guaranteed if the gap between the $k$th and $(k+1)$th eigenvalues of the
 Hamiltonian $H(X)$ outweighs the norm of the complicate term in the
 Hessian up to some constant and if the second-order derivatives of
 the exchange correlation energy is uniformly bounded from above.

  Although our conditions are restrictive for the convergence of
  the SCF iteration and they are almost never satisfied in reality, they still provide us some insights on the performance of
  the algorithm. Recently, numerical evidences show that the exact Hessian can
  speed up the convergence of the SCF iteration in the trust-region framework
  \cite{Wen2013}. Our analysis  has not covered the acceleration scheme using charge
  mixing since it is a fixed-point algorithm in terms of the charge density
  rather than minimizing the total energy functional.

\vspace{0.1cm}
\textbf{Acknowledgements} Z. Wen would like to thank Humboldt
 Foundation for the generous support, Prof. Michael Ulbrich for
hosting his visit at Technische Universit\"at
M\"unchen. The authors would like to thank Dr. Chao Yang for discussion on the
Kohn-Sham equation and are grateful
to  two anonymous referees
for their detailed and valuable comments and suggestions.


\begin{thebibliography}{10}

\bibitem{LeBris2005}
{\sc C.~L. Bris}, {\em Computational chemistry from the perspective of
  numerical analysis}, Acta Numer., 14 (2005), pp.~363--444.

\bibitem{Cances2000}
{\sc E.~Canc\`es}, {\em Scf algorithms for hartree-fock electronic
  calculations}, Lecture Notes in Chemistry, 74 (2000), pp.~17--43.

\bibitem{Cances2001}
\leavevmode\vrule height 2pt depth -1.6pt width 23pt, {\em Self-consistent
  field algorithms for kohn–sham models with fractional occupation numbers},
  Journal of Chemical Physics, 114(24) (2001), p.~10616–10622.

\bibitem{Cancesetal2000}
{\sc E.~Canc\`es and C.~L. Bris}, {\em Can we outperform the diis approach for
  electronic structure calculations?}, International Journal of Quantum
  Chemistry, 79(2) (2000), pp.~82--90.

\bibitem{CancesLeBris2000}
{\sc E.~Canc\`es and C.~L. Bris}, {\em On the convergence of {SCF} algorithms
  for the {Hartree-Fock} equations}, Math. Model. Numer. Anal., 34 (2000),
  pp.~749--774.

\bibitem{Cancesetal2003}
{\sc E.~Canc\`es, M.~Defranceschi, W.~Kutzelnigg, C.~L. Bris, and Y.~Maday},
  {\em Handbook of numerical analysis. Volume X: special volume: computational
  chemistry}, North-Holland, 2003, ch.~Computational quantum chemistry: a
  primer, pp.~3--270.

\bibitem{Cancesetal2008}
{\sc E.~Canc\`es and K.~Pernal}, {\em Projected gradient algorithms for
  hartree-fock and density matrix functional theory calculations}, Journal of
  Chemical Physics, 128(13) (2008), pp.~108--134.

\bibitem{EdelmanAriasSmith1998}
{\sc A.~Edelman, T.~Arias, and S.~Smith}, {\em The geometry of algorithms with
  orthogonality constraints}, SIAM J. Matrix Analysis Applications, 20(2)
  (1998), pp.~303--353.

\bibitem{GaoYangMeza2009}
{\sc W.~Gao, C.~Yang, and J.~Meza}, {\em Solving a class of nonlinear
  eigenvalue problems by {Newton's} method}, tech. rep., Lawrence Berkeley
  National Laboratory, 2009.

\bibitem{Golub1996}
{\sc G.~Golub and C.~V. Loan}, {\em Matrix Computaion}, The Johns and Hopkins
  University Press, 1996.

\bibitem{Kerker1981}
{\sc G.~P. Kerker}, {\em Efficient iteration scheme for self-consistent
  pseudopotential calculations}, Phys. Rev. B, 23 (1981), pp.~3082--3084.

\bibitem{KoutechyBonacic1971}
{\sc J.~Kouteck\'y and V.~Bonacic}, {\em On the convergence difficulties in the
  iterative {Hartree-Fock} procedure}, J. Chem. Phys., 55 (1971),
  pp.~2408--2413.

\bibitem{Kresse-1996}
{\sc G.~Kresse and J.~Furthmuller}, {\em Efficiency of ab-initio total energy
  calculations for metals and semiconductors using a plane-wave basis set},
  Computational Materials Science, 6 (1996), pp.~15--50.

\bibitem{Kudinetal2002}
{\sc K.~N. Kudin, G.~E. Scuseria, and E.~Canc\`es}, {\em A black-box
  self-consistent field convergence algorithm: One step closer}, Journal of
  Chemical Physics, 116(19) (2002), pp.~8255--8261.

\bibitem{Levitt2012}
{\sc A.~Levitt}, {\em Convergence of gradient-based algorithms for the
  hartree-fock equations}, ESAIM: Mathematical Modelling and Numerical
  Analysis, 46(6) (2012), pp.~1321--1336.

\bibitem{Martin2004}
{\sc R.~M. Martin}, {\em Electronic Structure: Basic Theory and Practical
  Methods}, Cambridge University Press, 2004.

\bibitem{Schneideretall2009}
{\sc R.~Schneider, T.~Rohwedder, A.~Neelov, Johannes, and Blauert}, {\em Direct
  minimization for calculating invariant subspaces in density functional
  computations of the electronic structure}, Journal of Computational
  Mathematics, 27(2/3) (2009), pp.~360--393.

\bibitem{Stewart1973}
{\sc G.~W. Stewart}, {\em Error bounds for approximation invariant subspace of
  closed linear operators}, SIAM Review, 15 (1973), pp.~27--64.

\bibitem{SzaboOstlund1996}
{\sc A.~Szabo and N.~S. Ostlund}, {\em Modern Quantum Chemistry: An
  Introduction to Advanced Electronic Structure Theory}, Dover, New York, 1996.

\bibitem{Wen2013}
{\sc Z.~Wen, A.~Milzarek, M.~Ulbrich, and H.~Zhang}, {\em Adaptive regularized
  self-consistent field iteration with exact hessian for electronic structure
  calculation}, SIAM Journal on Scientific Computing, 35(3) (2013),
  pp.~A1299--A1324.

\bibitem{YangGaoMeza2009}
{\sc C.~Yang, W.~Gao, and J.~Meza}, {\em On the convergence of the
  self-consistent field iteration for a class of nonlinear eigenvalue
  problems}, SIAM J. Matrix Analysis Applications, 30(4) (2009),
  pp.~1773--1788.

\bibitem{YangMezaLeeWang2009}
{\sc C.~Yang, J.~C. Meza, B.~Lee, and L.-W. Wang}, {\em {\rm KSSOLV}---a
  {MATLAB} toolbox for solving the {Kohn-Sham} equations}, ACM Trans. Math.
  Softw., 36 (2009), pp.~1--35.

\bibitem{YangMezaWang2007}
{\sc C.~Yang, J.~C. Meza, and L.~Wang}, {\em A trust region direct constrained
  minimization algorithm for the {Kohn-Sham} equation}, SIAM Journal of
  Scientific Computing, 29 (2007), pp.~1854--1875.

\end{thebibliography}
\end{document}